\documentclass[11pt]{article}
\usepackage{amsmath,amssymb,amsfonts,amsthm,epsfig}
\usepackage[dvipsnames]{xcolor}
\usepackage{bm,xspace}
\usepackage{tcolorbox}
\usepackage{cancel}
\usepackage{fullpage}
\usepackage{framed}
\usepackage{verbatim}
\usepackage{enumitem}
\usepackage{array}
\usepackage{multirow}
\usepackage{afterpage}
\usepackage{mathrsfs}
\usepackage{pifont} 
\usepackage{chngpage}
\usepackage[normalem]{ulem}
\usepackage{boxedminipage}
\usepackage{caption}
\usepackage{tikz}
\usepackage{mdframed}
\usepackage{setspace}
\usepackage{pras} %
\usepackage{multicol}
\usepackage{pgfplots}
\usepackage{dsfont}
\pgfplotsset{compat=1.18}

\newcommand{\DML}{D_{\textnormal{ML}}} 

\newcommand*\dif{\mathop{}\!\mathrm{d}}

\def\anonymizeymize{0}

\ifnum\anonymizeymize=1
\newcommand{\anonymize}[1]{}

\fi

\ifnum\anonymizeymize=0
\newcommand{\anonymize}[1]{#1}

\fi

\definecolor{darkgreen}{rgb}{0.0, 0.5, 0.0}

\def\comments{1}

\ifnum\comments=1
\newcommand{\pras}[1]{\textcolor{BrickRed}{\sf{[#1 --PR]}}}
\newcommand{\kangning}[1]{\textcolor{orange}{\sf{[#1 --KW]}}}
\newcommand{\moses}[1]{\textcolor{darkgreen}{\sf{[#1 --MC]}}}
\newcommand{\edit}[2]{\textcolor{red}{#1}}

\fi

\ifnum\comments=0
\newcommand{\pras}[1]{}
\newcommand{\kangning}[1]{}
\newcommand{\moses}[1]{}
\newcommand{\alex}[1]{}
\newcommand{\edit}[2]{#1}
\fi

\def\colorful{1}

\ifnum\colorful=1

\fi
\ifnum\colorful=0

\fi

\DeclareMathOperator{\rank}{rank}

\newcommand{\cg}{\succ}

\renewcommand{\hat}[1]{\widehat{#1}}

\renewcommand{\sf}[1]{\textsf{#1}}

\makeatletter
\newtheorem*{rep@theorem}{\rep@title}
\newcommand{\newreptheorem}[2]{
\newenvironment{rep#1}[1]{
 \def\rep@title{#2 \ref{##1}}
 \begin{rep@theorem}\itshape}
 {\end{rep@theorem}}}
\makeatother
\theoremstyle{plain}

\newreptheorem{theorem}{Theorem}

\makeatletter
\newtheorem*{rep@claim}{\rep@title}
\newcommand{\newrepclaim}[2]{
\newenvironment{rep#1}[1]{
 \def\rep@title{#2 \ref{##1}}
 \begin{rep@claim}\itshape}
 {\end{rep@claim}}}
\makeatother
\theoremstyle{plain}

\newrepclaim{claim}{Claim}

\newtheorem{Alg}{Algorithm}

\crefname{equation}{eq.}{eqs.}

\begin{document}

\definecolor{myblue}{rgb}{0.15, 0.1, 0.75}
\definecolor{mygreen}{rgb}{0.15, 0.55, 0.1}
\definecolor{mypink}{rgb}{0.75, 0.05, 0.55}

\hypersetup{
    linkcolor = mygreen,
    citecolor = myblue,
    urlcolor = mypink
}

\title{
Approximately Dominating Sets in Elections
}

\author{
\anonymize{
Moses Charikar\\\hspace{0pt}{{\sl Stanford University}}
\and Prasanna Ramakrishnan \\\hspace{0pt}{{\sl Stanford University}}
\and Kangning Wang\\\hspace{0pt}{{\sl Rutgers University}}
}
}
\anonymize{
{\let\thefootnote\relax\footnotetext{Emails: \texttt{moses@cs.stanford.edu}, ~\texttt{pras1712@stanford.edu}, ~\texttt{kn.w@rutgers.edu}.}}
}

\date{}

\pagenumbering{gobble}
\thispagestyle{empty}
\maketitle

\begin{abstract}

Condorcet's paradox is a fundamental result in social choice theory which states that there exist elections in which, no matter which candidate wins, a majority of voters prefer a different candidate. In fact, even if we can select any $k$ winners, there still may exist another candidate that would beat \emph{each} of the winners in a majority vote. That is, elections may require arbitrarily large \emph{dominating sets}.

We show that \emph{approximately} dominating sets of constant size always exist. In particular, for every $\eps > 0$, every election (irrespective of the number of voters or candidates) can select $O(\frac{1}{\eps^2})$ winners such that no other candidate beats each of the winners by a margin of more than $\eps$ fraction of voters.

Our proof uses a simple probabilistic construction using samples from a  \emph{maximal lottery}, a well-studied distribution over candidates derived from the Nash equilibrium of a two-player game. In stark contrast to general approximate equilibria, which may require support logarithmic in the number of pure strategies, we show that maximal lotteries can be approximated with \emph{constant} support size. These approximate maximal lotteries may be of independent interest.

\end{abstract}

\clearpage

\pagenumbering{arabic}

\section{Introduction}\label{sec:intro}

Elections are a pillar of collective decision making. While most commonly associated with representative democracy, the broad paradigm of a group of \emph{voters} choosing from a group of \emph{candidates} can model a wide range of scenarios, from organizations choosing new hires, judges choosing prizewinners, to friends choosing a restaurant or a movie. Recently, researchers have also found ways of using voting theory to inform machine learning, for which human preferences are now commonly used to fine-tune and align models \cite{christiano2017deep,DBLP:conf/nips/Ge0MPSV024,conitzer2024position,dai2024mapping}. 

Despite the enduring applicability of voting, its theoretical study has long been plagued with fundamental challenges. One of the earliest is \emph{Condorcet's paradox} \cite{condorcet1785essai}, dating back to the 18th century. It says that in elections where voters have ranked preferences over the candidates, a candidate that is preferred to any other by a majority of voters (called a \emph{Condorcet winner}) does not always exist. That is, no matter which candidate wins, a majority of voters may prefer another candidate. (In fact, this fraction of voters can be made arbitrarily close to 1.)

Ideally, we could hope to alleviate the paradox by showing that reasonable relaxations of a Condorcet winner \emph{do} always exist. One possibility is that if any winner might upset a large fraction of voters, perhaps we can achieve better outcomes by selecting a \emph{committee} of $k$ winners instead of just one. A few motivating examples include districts with multiple representatives, hiring for many openings, or to choose a shortlist of candidates like nominees or interviewees. From a theoretical perspective, committee selection has been the subject of a rich body of work \cite{fishburn1981analysis,fishburn1981majority,dummett1984voting, gehrlein1985condorcet,barbera2008choose,elkind2017properties}. (See \cite{faliszewski2017multiwinner,lackner2023multi} for more comprehensive surveys.)

A natural relaxation of a Condorcet winner in this setting is the notion of a \emph{dominating set}, for which every candidate outside the set is beaten by \emph{some} candidate in the set.

\begin{definition}[Dominating sets]\label{def:dom-set}
A committee $S$ is a \emph{dominating set} if for all candidates $a \notin S$, there exists $b\in S$ such that at least half\footnote{It is natural to use ``at least half'' instead of ``a majority,'' which would be slightly stronger. Otherwise a dominating set may trivially need to include all candidates if all pairs are tied.} of the voters prefer $b$ over $a$.   
\end{definition}

Unfortunately, minimal dominating sets in elections may be arbitrarily large, with size growing logarithmically in the number of candidates. This fact is a consequence of two classic results: McGarvey's Theorem \cite{mcgarvey1953theorem} which says that every tournament graph can be realized as the majority graph of some election (i.e., a tournament on candidates where  $a\to b$ if a majority of voters prefer $a \cg b$), and a result due to  Erd\H{o}s \cite{erdos1963problem} that there exist tournaments on $m$ vertices where all dominating sets are of size $\Omega(\log m)$. In fact, this is true for a random tournament. 

Instead, a promising recent approach by Elkind, Lang, and Saffidine \cite{elkind2011choosing,elkind2015condorcet} relaxes the notion of a dominating set by allowing the committee to dominate alternatives \emph{collectively}, treating it like a composite candidate which is simultaneously each voter's favorite candidate in the committee. 

\begin{definition}[Condorcet winning sets \cite{elkind2011choosing,elkind2015condorcet}]\label{def:cw-sets}
A committee $S$ is a \emph{Condorcet winning set} if 
for all candidates $a \notin S$, a majority of voters prefer some $b \in S$ over $a$. 
More generally, $S$ is a $\theta$-\emph{winning set} if for all candidates $a \notin S$, at least $\theta$ fraction of voters prefer some $b \in S$ over $a$. 
\end{definition}

Remarkably, a line of recent work proved that every election has a Condorcet winning set of \emph{constant} size. \cite{jiang2020approximately} implicitly showed that size 32 is sufficient, and this was improved to 6 \cite{DBLP:journals/corr/CharikarLRVW25}. These results also extend more generally to larger and smaller committee sizes. Asymptotically, there are always $(1 - O(\frac1k))$-winning committees of size $k$ \cite{jiang2020approximately}, and even for $k = 2$ there are $\theta$-winning committees with $\theta \approx 0.202$ \cite{DBLP:journals/corr/CharikarLRVW25}; bounded away from 0.\footnote{A new preprint \cite{nguyenfew} also improves these results, discussed further in \Cref{sec:related}.}

\vspace{11pt}

However, as \cite{elkind2011choosing,elkind2015condorcet} note, one strange feature of Condorcet winning sets is that they can be quite far from dominating sets. For example, a 1-winning set may exclude a candidate that would beat each member of the committee by a wide margin. To see this, consider an election where voters' first choices are distinct but their second choices are unanimous. The first choices would form a 1-winning set, but for any of these candidates, nearly all voters would prefer the second choice (which is in fact a Condorcet winner in the strongest sense). \cite{elkind2011choosing,elkind2015condorcet} also show empirically that the Condorcet winning sets and dominating sets in random elections are usually very different.

In light of this disparity, it is natural to wonder whether it is possible to have small committees that are close to dominating sets, without the need for collective domination. To this end, we raise the following question considering an  approximation of dominating sets.

\begin{question}\label{q:main}
A committee $S$ is an $\alpha$-\emph{dominating set} if for all candidates $a \notin S$, there exists $b\in S$ such that at least an $\alpha$ fraction of voters prefer $b$ over $a$. For each $\alpha \in [0, 1]$ what is the smallest $k_\alpha$ such that every election has an $\alpha$-dominating set of size at most $k_\alpha$? 
\end{question}

Note that dominating sets are precisely $\frac12$-dominating sets, meaning that $k_\alpha$ is unbounded for $\alpha \geq \frac12$ by \cite{mcgarvey1953theorem,erdos1963problem}. Thus, the region of interest is $\alpha \in [0, \frac12)$.

Using existing work, we can piece together a little more of the picture. Generalizing McGarvey's Theorem, a line of work in combinatorics \cite{stearns1959voting,erdos1964representation,alon2002voting,alon2006dominating} has studied what kinds of tournament graphs are realizable by elections with few voters, or while ensuring edges represent large margins. One can think of an election just as a collection of linear orders, which naturally lends them to a combinatorial treatment. These works are also motivated by the connection between voting paradoxes and \emph{intransitive dice}, which we explore further in \Cref{sec:related}.  

Alon \cite{alon2002voting} showed that every $m$-vertex tournament graph can be realized by an election where for all edges $a\to b$, $\frac12 + \Omega(\frac{1}{\sqrt{m}})$ fraction of voters prefer $a$ over $b$. This result immediately implies that as $\alpha$ tends to $\frac12$, $k_\alpha$ tends to infinity. In particular, $k_{\frac12 - \eps} \geq \Omega(\log \frac1\eps)$. A result of Alon, Brightwell, Kierstead, Kostochka, and Winkler \cite{alon2006dominating} implies a sharper quantitative bound. They prove the existence of elections with $n$ voters, where $n$ is odd, and every dominating set has size $\Omega(n/\log n)$ (which, they show, is tight up to a $\log^2 n$ factor). 
It follows that $k_{\frac12 - \eps} \geq \Omega\Big(\frac{1}{\eps\log \frac1\eps}\Big)$. This bound treats their result as a black box; in \Cref{sec:lbs} we show that their construction implies $k_{\frac12 - \eps} \geq \Omega(\frac{1}{\eps})$.

On the other hand, the recent work on $\theta$-winning sets can be used to obtain positive results for small $\alpha$. It is not hard to see that while any $\alpha$-dominating set is always an $\alpha$-winning set, a $\theta$-winning set of size $k$ is at the very least a $\frac{\theta}{k}$-dominating set. With this, \cite{DBLP:journals/corr/CharikarLRVW25} imply that for for $\alpha \leq 0.101$, $k_\alpha = 2$. Though they give stronger guarantees on $\theta$-winning sets for larger committee sizes, they do not yield stronger guarantees for $\alpha$-dominating sets (due to the factor $k$ dilution).

As such, we have no non-trivial results for larger $\alpha$, let alone $\alpha$ close to $\frac12$, unless we restrict the space of elections in consideration. For example, \cite{alon2006dominating} consider elections of \emph{quality} $\eps$, for which the fraction of voters that prefer any $a$ over any $b$ is outside $(\frac12 - \eps, \frac12 + \eps)$ (i.e., all pairwise margins are reasonably large). In these elections, $(\frac12 -\eps)$-dominating sets and ($\frac12$-)dominating sets are the same. \cite{alon2006dominating} show that restricting to elections of quality $\eps$, we can find dominating sets of size $O(\frac{1}{\eps^2}\log \frac1\eps)$.

\vspace{11pt}

In this work, we show that without any restrictions, $k_\alpha$ is finite for all $\alpha \in [0,\frac12)$. In particular, $k_{\frac12 - \eps} \leq O(\frac1{\eps^2})$.

\begin{theorem}\label{thm:near-dom-sets}
For all $\eps > 0$, every election has a $(\frac12 - \eps)$-dominating set $S$ of at most $(1 + o(1))\frac{\pi}{8\eps^2}$ candidates. In fact, there is a distribution $D$ over $S$ such that for all candidates $a$, in expectation, at most $\frac12 + \eps$ fraction of voters prefer $a$ over $b \sim D$.
\end{theorem}

A qualitative interpretation of \Cref{thm:near-dom-sets} is that we can always choose a small committee that only excludes a candidate if it includes one that is \emph{effectively} just as good in the eyes of the voters (for example, to distinguish between the two, one would need to sample a large number of voters' preferences). It is also worth noting that for the stronger version of the result, the distribution $D$ is quite simple --- it is just a uniform distribution over a multiset of candidates, whose size is still $(1 + o(1))\frac{\pi}{8\eps^2}$. We can interpret the result as showing that any excluded candidate is barely better than the \textit{average} candidate in the multiset. This framing suggests a plausible normative justification for selecting this set of candidates: every other candidate is reasonably excluded, and each candidate in the set is (literally) pulling its weight.

\Cref{thm:near-dom-sets} can also be thought of as strengthening \cite{alon2006dominating}'s result on elections of quality $\eps$ in three ways. First, our result applies to \emph{all} elections, by relaxing dominating sets to $(\frac12 - \eps)$-dominating sets instead of restricting to elections of quality $\eps$. (Since $(\frac12 -\eps)$-dominating sets and dominating sets are the same in elections of quality $\eps$, their result is a corollary of \Cref{thm:near-dom-sets}.) Second, rather than having to choose a candidate $b \in S$ adaptively against each $a \notin S$, our result even works with a fixed distribution $D$ over the candidates in $S$. Third, our quantitative guarantee is better by a $\log \frac1\eps$ factor.

\subsection{Approximating Maximal Lotteries} 

The heart of our approach is the concept of a \emph{maximal lottery} \cite{kreweras1965aggregation,fishburn1984probabilistic}, a distribution over candidates which is the mixed-strategy Nash equilibrium of the following simple two-player zero-sum game. Given a fixed election, Alice chooses a candidate $a$, Bob chooses a candidate $b$, and then Alice's payoff is half the vote margin of $a$ over $b$ (i.e., $\tfrac1n|a\cg b| - \frac12$ where $\tfrac1n|a\cg b|$ denotes the fraction of voters that prefer $a$ over $b$), and Bob's payoff is the negation.

Maximal lotteries have a rich literature spanning a wide variety of applications and disciplines. As \cite{brandt2017rolling} details, they were first defined by Kreweras in 1965 \cite{kreweras1965aggregation}, studied independently in detail by Fishburn \cite{fishburn1984probabilistic}, but then rediscovered repeatedly in political science \cite{felsenthal1992after}, economics \cite{laffond1993bipartisan}, math \cite{fisher1995tournament},  and computer science \cite{rivest2010optimal}. More recently, maximal lotteries have also played a role in recent breakthroughs in \emph{metric distortion} \cite{charikar2024breaking}, and have garnered significant interest for applications in reinforcement learning and AI alignment \cite{lanctot2023evaluating,wang2023rlhf,munos2023nash,swamy2024minimaximalist, rosset2024direct,gao2024rebel,ye2024theoretical,maura2025jackpot}.

By the symmetry of the game, a maximal lottery has\edit{}{are mixed strategies which} non-negative expected payoff against any pure strategy the opponent chooses. Equivalently, in the voting parlance, we have \Cref{fact:rcw} below, which shows that a maximal lottery can be thought of as a kind of \emph{randomized} Condorcet winner.

\begin{fact}\label{fact:rcw}
In any election, there is a maximal lottery distribution $\DML$ such that for all candidates $a$, in expectation, at most half the voters prefer $a$ over $b\sim \DML$. 
\end{fact}

With this game-theoretic correspondence in mind, what  \Cref{thm:near-dom-sets} shows is that we can \emph{approximate} maximal lotteries with constant support. More precisely, we can construct a mixed strategy with support size $O(\frac{1}{\eps^2})$ that guarantees expected payoff at least $-\eps$ in the aforementioned game.

This result contrasts sharply with the approximability of general Nash equilibria with small support, another intensively studied subject \cite{althofer1994sparse,lipton2003playing,feder2007approximating,daskalakis2009note,babichenko2014simple,rubinstein2017settling}. Alth\"ofer \cite{althofer1994sparse} showed that in two-player zero-sum games with $m$ pure strategies, there exists an $\eps$-Nash equilibrium ($\eps$-NE) with support size $O(\frac{1}{\eps^2}\log m)$, and Lipton, Markakis, and Mehta \cite{lipton2003playing} generalized this to all two-player games. These results are quantitatively the same as our guarantee for approximate maximal lotteries, but with an additional $\log m$ factor. In fact, on the surface our approach is similar to that of  \cite{althofer1994sparse} and \cite{lipton2003playing}: sample repeatedly from the Nash equilibrium and argue using concentration bounds that with positive probability, the resulting empirical distribution is an $\eps$-NE.

However, for $\eps$-NE, we cannot get rid of the $\log m$ factor. Feder, Nazerzadeh, and Saberi \cite{feder2007approximating} show that the $O(\frac{1}{\eps^2}\log m)$ bound of \cite{althofer1994sparse,lipton2003playing} is optimal. A corollary of \cite{lipton2003playing} is that one can find an $\eps$-NE in quasi-polynomial time ($m^{O(\log m)}$) via exhaustive search. Rubinstein \cite{rubinstein2017settling} showed that this is also essentially optimal, subject to reasonable complexity-theoretic assumptions on the hardness of \textsf{PPAD}. 

It is worth noting that the dependence on $\log m$ is not because of games that are totally unrelated to the motivations of \Cref{q:main} --- it is \emph{precisely} because tournaments can require dominating sets of $\Omega(\log m)$ size. Consider a game where each player chooses a vertex of the tournament, and the direction of the edge between them determines the winner ($a\to b$ means $a$ wins, if they pick the same vertex it is a tie). If all dominating sets have size greater than $k$, then clearly no strategy that only randomizes between at most $k$ vertices can be close to an equilibrium. Indeed, \cite{feder2007approximating}'s lower bound uses a game with a random 0-1 payoff matrix, essentially equivalent to Erd\H{o}s' random tournaments \cite{erdos1963problem}. In a similar fashion, \cite{anbalagan2015large} construct graphs with no small dominating sets and large girth to show that approximate \emph{well-supported Nash equilibria} also require large supports. %

All together, these connections make it somewhat surprising that maximal lotteries can be approximated with constant support size, without the $\log m$ dependence one might expect (where $m$ is the number of candidates). It is also fortuitous that the construction and proof are simple, and it would be interesting to study whether similar techniques can extend to other specific approximate Nash equilibria of interest. Moreover, given the vast variety of research developing and using maximal lotteries, approximate maximal lotteries may be of independent interest.

\subsection{Further Related work}\label{sec:related}

\paragraph{Intransitive dice.} As mentioned earlier, voting paradoxes are closely related to the counterintuitive phenomenon of \emph{intransitive dice}, the fact that it is possible to assign numbers to the faces of dice such that the first die is more likely to roll higher than the second, the second  is more likely to roll higher than the third, and so on, until finally the last die is more likely to roll higher than the first. Using these dice one can trick an unsuspecting victim into a losing game, offering them an apparent first mover's advantage in picking a die, and then choosing the die that beats their choice. In fact, the winning probability can be made close to $\frac34$ \cite{usiskin1964max,trybula1965paradox}, so rolling the dice a few times, the trickster is almost certain to win more often. 

Intransitive dice were popularized by Martin Gardner in 1970 \cite{gardner1970paradox}, and they raise a broad range of interesting mathematical questions that have been studied extensively \cite{trybula1961paradox,usiskin1964max,trybula1965paradox, conrey2016intransitive, bru2018dice, buhler2018maximally, akin2019generalized, hkazla2020probability,komisarski2021nontransitive, kirkegaard2022emerging, polymath2022probability, kim2024balanced, sah2024intransitive}. 

To see the connection to voting, consider an election where each die is a candidate, each possible sample (rolling all dice together) is a voter $v$, and $v$'s preference is the order of the dice in the sample. The intransitivity of the dice is equivalent to the absence of a Condorcet winner in the corresponding election. Using this correspondence, the fact that dominating sets in elections may be arbitrarily large implies that one can design dice to \emph{simultaneously} trick any number of victims.\footnote{As \cite{alon2006dominating} vividly write, \textit{``much as a chess master puts on a simultaneous exhibition.''}} \cite{alon2006dominating} are motivated in part by understanding the tradeoff between the number of victims and the number of faces on each die. 

Similarly, \Cref{thm:near-dom-sets} has an interesting interpretation in terms of dice. If a trickster wants to handle many victims at once, but beat each of them with probability $50.1\%$ instead of $50\%$, then suddenly the number of victims is bounded, even if any number and manner of dice can be made. If sufficiently many would-be victims coordinate carefully, as they roll the dice, at least one of them must be giving the trickster a run for their money. In fact, if playing many against one seems too unfair, the would-be victims can propose that for each roll, one of them will be chosen at random to play (allowing duplicate/shared dice) and they still have effectively even odds.

\paragraph{Other Condorcet committees.} Through \Cref{def:dom-set,def:cw-sets}, we have seen that there are a variety of possible ways of generalizing the notion of Condorcet winners to committee selection, and it is helpful to contextualize these with a discussion of other proposals that have been considered in the literature. Often, these alternatives are not guaranteed to exist (like Condorcet winners) or may need to an arbitrarily large number of candidates (like $\alpha$-dominating sets for $\alpha \geq \frac12$). Thus, it is exciting to find a new definition that avoids these drawbacks. %

Fishburn \cite{fishburn1981analysis,fishburn1981majority} defined a \emph{majority committee} to be one that is preferred to any other committee of the same size by a majority of voters. This notion requires voters to provide preferences over committees (rather than just candidates), and may not always exist, akin to Condorcet's paradox. Fishburn \cite{fishburn1981majority} studied restricted preference structures that ensure the existence of majority committees. \cite{good1971note,smith1973aggregation} introduced the \emph{Smith set}, which is the minimal set of candidates $S$ such that a majority of voters prefers any $a \in S$ over any $b \notin S$. The Smith set is used to define a stronger version of the \emph{Condorcet criterion}. %

Researchers have also studied the support of maximal lotteries, called the \emph{bipartisan set} \cite{laffond1993bipartisan}. (Note that the bipartisan set is always a subset of the Smith set.) Since the $(\frac12 - \eps)$-dominating sets we construct are the support of approximate maximal lotteries (and obtained by sampling from a maximal lottery), they can be thought of as a kind of sparsified bipartisan set. 

A variety of other sets can be defined by applying graph-theoretic concepts to the tournament graph defined by pairwise majority votes between candidates. We refer the reader to \cite{brandt2016tournament} for a detailed survey.

\paragraph{New work on $\theta$-winning sets.} 
 Nguyen, Song, and Lin \cite{nguyenfew} very recently posted a new preprint, improving \cite{DBLP:journals/corr/CharikarLRVW25}'s bounds on $\theta$-winning sets for all committee sizes $k\geq 2$. %

Their main result is that all elections have Condorcet winning sets of size at most 5. In terms of $\alpha$-dominating sets, their results imply that $k_\alpha = 2$ for $\alpha \leq \frac18 = 0.125$. Intriguingly, unlike \cite{DBLP:journals/corr/CharikarLRVW25}, their result \textit{does} give something marginally stronger for committees of size $3$, showing that $k_\alpha \leq 3$ for $\alpha \leq 0.128$, but the improvements go no further. As we note in \Cref{rmk:small-k}, if we work out our bounds more precisely for small values of $k$, we would get a matching bound for $k = 2$ (which is curious given the difference in techniques), and an improved bound for $k \geq 3$. 

\paragraph{Processes for approximating maximal lotteries.} Brandl and Brandt \cite{brandl2024natural} propose a simple and elegant scheme that approximates maximal lotteries with only lightweight preference elicitation. They imagine an urn  filled with many balls, each arbitrarily labeled with a candidate. In every round two balls are drawn uniformly at random, and a randomly selected voter is compares the corresponding candidates. The ball representing the losing candidate is replaced with a copy of the winning one. To maintain ergodicity, there is also a small probability that a random ball is randomly relabeled. \cite{brandl2024natural} show that with sufficiently many balls, the distribution of labels converges exponentially fast to a maximal lottery.

\subsection{Concurrent Work}

Independent to our work, a new preprint by Bourneuf, Charbit, and Thomass\'e \cite{bourneuf2025dense} also shows that elections have approximately dominating sets of constant size, specifically $(\frac12 - \eps)$-dominating sets of size $O(\frac{1}{\eps^2})$ --- the same as the first part of \Cref{thm:near-dom-sets}. Intriguingly, our results are proved via very different techniques. \cite{bourneuf2025dense}'s approach can be thought of as extending \cite{alon2006dominating}, which argues that majority graphs derived from elections of quality $\eps$ must have bounded VC-dimension. \cite{bourneuf2025dense} adapt majority graphs to allow two kinds of edges (differentiating large and small margins between candidates) and show that an appropriate generalization of the VC-dimension is still bounded using the \textit{dense neighborhood lemma}, a tool whose versatility they showcase through a wide range of applications.

These differing approaches offer distinct advantages. Our approach is somewhat more direct, and gives a sharper quantitative guarantee that also extends to distributions over candidates with small support. On the other hand, \cite{bourneuf2025dense} show that their techniques can be applied to a broader suite of combinatorial structures beyond elections. It would be exciting to understand whether the best of both worlds is possible, and more generally to explore deeper connections between the two approaches.

\section{Preliminaries and Notation}\label{sec:prelims}

\paragraph{Elections.} Formally, an \emph{election} (or \emph{preference profile}) is defined by a tuple $(V, C, \cg_V)$ where $V$ is a set of $n$ voters, $C$ is a set of $m$ candidates, $\cg_V = \{\cg_v\}_{v\in V}$ is a set of linear $n$ orders over $C$, one corresponding to each voter $v \in V$. We say that $a\cg_v b$ if voter $v$ prefers candidate $a$ over candidate $b$. We use 
$$\tfrac1n|a\cg b| := \frac1n\sum_{v\in V} \mathds{1}[a\cg_v b]$$
to denote the fraction of voters that prefer $a$ over $b$. Note that $\mathds{1}[a\cg_v a] = 0$.

\paragraph{Distributions.} We use $D$ (and variants) to denote distributions over candidates, and $x\sim D$ to say that $x$ is sampled from $D$. With a set $S$ we also write $x \sim S$ to say that $x$ is sampled from the \emph{uniform} distribution over $S$. (Specifically, $x\sim [0, 1]$ denotes sampling a uniform real number between $0$ and $1$, and $v \sim V$ denotes sampling a uniformly random voter.) 

\paragraph{Rank.} Lastly, we introduce the notion of \emph{rank}, borrowing from \cite{DBLP:journals/corr/CharikarLRVW25}, a tool which helps make the analysis more simple and intuitive. Given a distribution $D$ over candidates, voter $v$'s rank of $a$ with respect to $D$ is 
$$\rank_v(a; D) := \Prx_{b\sim D} [a\cg_v b].$$
Intuitively, the rank gives a quantitative measure of how much a voter $v$ likes a candidate $a$ by comparing it to the distribution $D$. A helpful fact is that the average rank of $a$ over the voters is precisely the expected fraction of voters that prefer $a$ over $b\sim D$. That is, 
\begin{equation}\label{eq:rank-ev}
\frac1n\sum_{v\in V}\rank_v(a; D) = \Ev_{b\sim D}\big[\tfrac1n|a \cg b|\big].
\end{equation}
Crucially, when $D$ is a maximal lottery, \Cref{fact:rcw} tells us that the average rank of each candidate $a$ is at most $\frac12$.

\section{Proof of \texorpdfstring{\Cref{thm:near-dom-sets}}{Theorem 1}}\label{sec:mainthm}

Before diving into the proof, we first provide some high-level intuition for the analysis, particularly for why it is possible to avoid a log factor blow up. Our goal is to construct a distribution with small support that approximates a maximal lottery. This distribution will simply be the empirical distribution of a number of samples from the maximal lottery. To show that the empirical distribution is not too different from a maximal lottery, we want to argue that for each candidate $a$, the average voter's rank of $a$ with respect to the maximal lottery and with respect to the empirical distribution are not too different. We can make such an argument using concentration inequalities, but if we need to union bound over all $m$ candidates, we will incur a $\log m$ blow up in the number of samples. What saves us is that each voter's ranks of the candidates always follow the linear order of that voter's preferences, so if one candidate's rank changes only slightly, it limits other candidates with close ranks from changing significantly. 
Heuristically, if we want to show that all ranks change by at most $O(\eps)$, it is really sufficient to show this for $O(1/\eps)$ evenly distributed ranks, since this will extend to all ranks that are sandwiched in between. This argument would replace the $\log m$ with a $\log \frac1\eps$. 

As it turns out, this measurement of the most that the ranks with respect to a particular voter can change is characterized by classical statistical machinery that measures the difference between the CDFs of a distribution and its empirical distribution. We can appeal directly to this machinery to get a sharper bound than the heuristic argument above, shaving even the $\log \frac1\eps$ factor. We start the proof with \Cref{lem:DKW}, which distills the technical tool we borrow.

\begin{lemma}\label{lem:DKW}
Given $X_1, \dots, X_k \in [0,1]$, let  $\hat{F}(r) := \frac1k\sum_{i=1}^k \mathds{1}[X_i<r]$. Define 
$$\delta(k) := \Ev_{X_1, \dots, X_k\sim [0, 1]}\left[\sup_{r\in[0,1]}\left(\hat{F}(r) - r\right) \right].$$
Then 
$$\delta(k) \leq (1 + o(1))\sqrt{\frac{\pi}{8k}}.$$
\end{lemma}

Note that $F(r) = r$ is the CDF of the uniform distribution, whereas $\hat{F}(r)$ is the CDF of the empirical distribution of the samples $X_1, \dots, X_k$. The largest difference between these two distributions (whose expectation is $\delta(k)$) is closely connected to statistical tests like the Kolmogorov--Smirnov test, and is well understood. Massart's \cite{massart1990tight} improvement of the Dvoretzky--Kiefer--Wolfowitz inequality  gives an essentially tight characterization of the CDF of the largest difference, and \Cref{lem:DKW} is a straightforward consequence of this result.

\begin{proof}[Proof of \Cref{lem:DKW}]
Massart \cite{massart1990tight} showed that for $\lambda \leq \min\Big(\sqrt{\frac{\ln 2}{2}}, 1.09k^{-1/6}\Big)$,
$$\Prx_{X_1, \dots, X_k\sim [0, 1]}\left[\sqrt{k}\sup_{r\in[0,1]}\left(\hat{F}(r) - r\right) > \lambda \right]\leq  e^{-2\lambda^2}.$$
It follows that 
\begin{align*}
\sqrt{k}\cdot \delta(k) &\leq \int_{0}^\infty \Prx_{X_1, \dots, X_k\sim [0, 1]}\left[\sqrt{k}\sup_{r\in[0,1]}\left(\hat{F}(r) - r\right) > \lambda \right]\dif \lambda \\
&\leq 1.09k^{-1/6} + \int_0^\infty e^{-2\lambda^2} \dif \lambda \\
&= 1.09k^{-1/6} + \sqrt{\frac{\pi}{8}}.
\end{align*}
Hence, $\delta(k) \leq (1 + o(1))\sqrt{\frac{\pi}{8k}}.$
\end{proof}

We note that, along the lines of the earlier heuristic argument, one can also prove a weaker bound of $\delta(k) = O\Big(\sqrt{\tfrac{\log k}{k}}\Big)$ with less heavy machinery. We can discretize the interval $[0,1]$ into grid points $r = \frac{1}{k}, \frac{2}{k}, \dots, \frac{k}{k}$, and argue via Hoeffding's inequality that with high probability, $\hat{F}(r)$ is within $O\Big(\sqrt{\tfrac{\log k}{k}}\Big)$ of $r$ for all of these points. Due to the fine discretization and monotonicity of $\hat{F}(\cdot)$, this extends to all $r \in [0, 1]$ with at most an additional $\frac1k$ error.

\vspace{11pt}

We are now ready to prove our main theorem, restated more formally below.

\begin{reptheorem}{thm:near-dom-sets}
In any election, there exists a distribution $D$ whose support size is $(1 + o(1))\frac{\pi}{8\eps^2}$ such that $$\max_{a \in C} \Ev_{b\sim D}\big[\tfrac1n|a \cg b|\big] \leq \frac12 + \eps.$$
\end{reptheorem}

\begin{proof}
Fix an election and let $\DML$ denote the distribution of a maximal lottery. Note that the key property of $\DML$ (\Cref{fact:rcw}) is that for all candidates $a \in C$, 
\begin{equation}\label{eq:ml}
\frac1n\sum_{v\in V}\rank_v(a; \DML) = \Ev_{b\sim \DML}\big[\tfrac1n|a \cg b]\big] \leq \frac12.
\end{equation}

Now suppose we sample $b_1, \dots, b_k \sim \DML$, where $k$ will eventually be $(1 + o(1))\frac{\pi}{8\eps^2}$. We claim that with positive probability, the empirical distribution $\hat{D}$ which samples uniformly from $b_1, \dots, b_k$ satisfies the property required by the theorem. The intuition is to argue that in expectation, for each candidate $a$ and the average voter $v$, $\rank_{v}(a ;\hat{D})$ is not too much larger than $\rank_{v}(a ; \DML)$, and so $\hat{D}$ will satisfy an approximate version of \Cref{eq:ml}. 

From the perspective of each voter $v$, we can view samples from any distribution $D$ over candidates in the following way. We split the interval $[0, 1]$ into disjoint segments corresponding to candidates, and the length of the segment for each candidate $a$ is the probability mass of $a$ in $D$. The segments are ordered by $v$'s preference, with more preferred candidates closer to $1$ and less preferred closer to $0$. Note that $\rank_v(a;D)$ is precisely the minimum of the segment corresponding to $a$. 

To sample $b\sim D$, we can sample a uniform $X\sim [0, 1]$, and take the candidate corresponding to the interval that $X$ lands in. In other words, $b$ is the candidate such that $\rank_v(b;D)$ is maximal but at most $X$. We will view $b_1, \dots, b_k$ as sampled via this process
(where $D\leftarrow \DML$), with corresponding $X_1, \dots, X_k \sim [0, 1]$. (See \Cref{fig:rank-change}.)

\begin{figure}[htbp]
  \centering
  \includegraphics[width=0.8\textwidth]{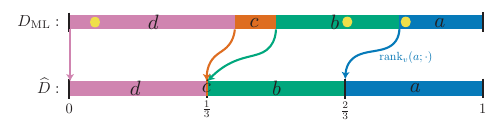}
  \caption{A visual of how the ranks change from $\DML$ to $\hat{D}$ for a particular voter $v$ with ranking $a\cg b\cg c \cg d$. $\DML$ chooses $a, b, c, d$ with probability $0.2,0.3,0.1,0.4$ respectively, and these correspond to the widths of each segment in the top line. The yellow points represent the samples $X_1, X_2, X_3$, and the arrows indicate how the ranks of each candidate change. For example, $\rank_v(a;\DML) = 0.8$ and $\rank_v(a;\hat{D}) = \frac23$.}
  \label{fig:rank-change}
\end{figure}

In this view, the probability that $v$ prefers $a$ over $b\sim \hat{D}$ is precisely the fraction of samples $X_1, \dots, X_k$ that are less than $\rank_v(a;\DML)$. That is,
\begin{equation}\label{eq:emperical-rank}
\rank_{v}(a;\hat{D}) = \frac1k\sum_{i=1}^k \mathds{1}[X_i<\rank_v(a;\DML)].
\end{equation}
Let 
$$\delta_v := \max_{a\in C}\left(\rank_v(a;\hat{D}) - \rank_v(a; \DML)\right)$$
denote the largest increase in rank (with respect to $v$) that any candidate experiences in changing from $\DML$ to $\hat{D}$. Using \Cref{eq:emperical-rank}, 
$$\delta_v \leq \sup_{r \in [0, 1]} \left(\frac1k\sum_{i=1}^k \mathds{1}[X_i<r] - r\right)$$
and using $\delta(k)$ as defined in \Cref{lem:DKW}, we have
$$\Ev_{b_1, \dots, b_k\sim \DML}[\delta_v] \leq \delta(k).$$
By linearity of expectation, we can also bound the \emph{average} increase in ranks across the voters.
$$\Ev_{b_1, \dots, b_k\sim \DML}\left[\frac1n\sum_{v\in V}\delta_v\right] \leq \frac1n\sum_{v\in V} \Ev_{b_1, \dots, b_k\sim \DML}\left[\delta_v\right]\leq \delta(k).$$
By the probabilistic method, there exist $b_1, \dots, b_k$ such that $\frac1n\sum_{v\in V}\delta_v \leq \delta(k)$. Since $\rank_v(a;\hat{D}) \leq \rank_v(a;\DML) + \delta_v$ for each voter $v$ and each candidate $a$, we can use this to bound the average rank of any candidate $a$ with respect to $\hat{D}$, which is precisely the expected fraction of voters that prefer $a$ over $b \sim \hat{D}$. More explicitly,
\begin{align*}
\Ev_{b\sim \hat{D}}[\tfrac1n|a \cg b|] &= \frac1n\sum_{v\in V}\rank_v(a;\hat{D}) \\
&\leq \frac1n\sum_{v\in V} (\rank_v(a;\DML) + \delta_v)\\
&\leq \frac1n\sum_{v\in V} \rank_v(a;\DML) + \frac1n\sum_{v\in V} \delta_v\\
&\leq \frac12 + \delta(k) \\
&\leq \frac12 + (1 + o(1))\sqrt{\frac{\pi}{8k}}. 
\end{align*}
Therefore, taking $k = (1 + o(1))\frac{\pi}{8\eps^2}$, $\hat{D}$ satisfies the conditions of the theorem.
\end{proof}

\begin{remark}\label{rmk:small-k}
The argument in general shows that $k_\alpha \leq k$ for $\alpha \leq \frac12 - \delta(k)$, and this can be used to obtain sharper results for small values of $k$ as well. One can write explicitly that 
$$\delta(k) = k! \int_{0\leq x_1\leq \cdots \leq x_k \leq 1} \max_{j= 1}^k \left(\frac{j}{k} - x_j\right) \dif x_1 \cdots \dif x_k.$$
Evaluating these integrals for $k = 2,3,4$, we can get that $\delta(2) = \frac38, \delta(3) = \frac{17}{54}, \delta(4) = \frac{71}{256}$, which shows that $k_\alpha = 2$ for $\alpha \leq \frac18 = 0.125$,  $k_\alpha \leq 3$ for $\alpha \leq \frac{5}{27} \approx 0.185$, and $k_\alpha \leq 4$ for $\alpha \leq \frac{57}{256} \approx 0.222$. In \Cref{fig:exact}, we plot the result we would get by computing $\delta(k)$ exactly, and compare with the estimate given in \Cref{thm:near-dom-sets}. 
\end{remark}

\begin{figure}
\centering
\begin{tikzpicture}
\begin{axis}[
    width=14cm,
    height=7cm,
    xlabel={$k$},
    ylabel={$\alpha$},
    ymin=0, ymax=0.5,
    xmin=0, xmax=50,
    grid=major,
    legend style={at={(1.125,0.585)}, anchor=north},
    legend cell align=left
]

\addplot[
    only marks,
    mark=*,
    color=blue,
] coordinates {
(2, 0.125)
(3, 0.185185185185185185185185185185)
(4, 0.22265625)
(5, 0.24896)
(6, 0.268775720164609053497942386831)
(7, 0.284418664234897267052236495241)
(8, 0.297186374664306640625)
(9, 0.307871347506352458297578577472)
(10, 0.316989216)
(11, 0.324892179511938219415300285254)
(12, 0.331830263537543700464868160341)
(13, 0.337986618732595379085939763403)
(14, 0.343498949129160032747076956116)
(15, 0.348473090495092387678030110586)
(16, 0.352991929778978807963341068898)
(17, 0.357121446447010960703164460218)
(18, 0.360914913916864092100298572096)
(19, 0.364415889448488200465716970532)
(20, 0.3676603853499774983984375)
(21, 0.370678474205886111613584562403)
(22, 0.373495494821360367241051234313)
(23, 0.376132971327249644639819081626)
(24, 0.378609322834946103092444367465)
(25, 0.380940417875094817334283646195)
(26, 0.383140012256036489382284159054)
(27, 0.385220098280020042259092897139)
(28, 0.387191185797439537773269575433)
(29, 0.389062530301916299891044464237)
(30, 0.390842319482782906634858092693)
(31, 0.39253782690039059671538842675)
(32, 0.394155539427035306783599323023)
(33, 0.395701263592960017286891026637)
(34, 0.397180214848116881709244302618)
(35, 0.398597092894768133829978067181)
(36, 0.399956145591686844509924863827)
(37, 0.401261223426157405875784738253)
(38, 0.402515826157846040263502828391)
(39, 0.403723142931612642847520485039)
(40, 0.404886086914313817677455300465)
(41, 0.406007325318597300947693239473)
(42, 0.407089305523345367799325320532)
(43, 0.408134277877267412610433996906)
(44, 0.409144315672675474013379813798)
(45, 0.410121332695717611098033441051)
(46, 0.411067098693443330290965246149)
(47, 0.411983253044040601667746742245)
(48, 0.412871316872075719982498234091)
(49, 0.413732703813745427287358392964)
(50, 0.414568729606560506253809399864)
};\addlegendentry{$\frac12 - \delta(k)$}

\addplot[
    domain=2:50,
    samples=100,
    thick,
    dashed,
    color=red,
] {0.5 - sqrt(pi / (8 * x))};
\addlegendentry{$\frac{1}{2} - \sqrt{\frac{\pi}{8k}}$}

\end{axis}
\end{tikzpicture}
\caption{A blue point at $(k^*, \alpha^*)$ implies that every election has an $\alpha^*$-dominating set of at most $k^*$ candidates (i.e, $k_\alpha \leq k^*$ for $\alpha \leq \alpha^*$). These points are derived by evaluating $\delta(k)$ directly, and the red line compares to the estimate from \Cref{lem:DKW} that $\delta(k) \approx \sqrt{\frac{\pi}{8k}}$.}\label{fig:exact}. 
\end{figure}

\begin{remark}
While \Cref{thm:near-dom-sets} is stated as an existence result, it is not hard to make it algorithmic and find small approximately dominating sets in expected polynomial time. Firstly, since the maximal lottery is the equilibrium of a two-player $m\times m$ zero-sum game, it can be computed in polynomial time by solving a linear program. The proof samples $k$ times from the maximal lottery, and argues that the result is a $(\frac12 - \delta(k))$-dominating set with positive probability since the expectation of $Z := \frac{1}{n} \sum_{v \in V}\delta_v$ is at most $\delta(k)$. We can instead use Markov's inequality to argue that $Z$ is at most $(1+\eta) \delta(k)$ with probability at least $\frac{\eta}{1+\eta}$. Taking $\eta = \frac1m$ for example, by repeating the sampling process until it succeeds, we can find a $(\frac12 - (1 + o(1))\delta(k))$-dominating set in expected polynomial time.
\end{remark}

\section{Discussion}\label{sec:discussion}

Our work, taken together with the recent line of work on Condorcet winning sets \cite{jiang2020approximately,DBLP:journals/corr/CharikarLRVW25} continues to demonstrate the surprising ``power of many choices'' in voting (much like other settings in theoretical computer science like list decoding and load balancing). The fact that small committees can satisfy such \textit{diverse} desiderata is particularly striking. Moreover, our work also shows how randomized solutions like maximal lotteries, which can sometimes be shied away from due to their nondeterminism,  can be powerful tools for understanding deterministic criteria via the probabilistic method.

We briefly conclude with a couple of directions for future work.

\paragraph{Sharper asymptotic bounds.} While our results show that $k_\alpha$ is finite for all $\alpha < \frac12$, it would be interesting to precisely pin down the asymptotic tradeoff between $\alpha$ and $k_\alpha$. In particular, our best bounds are $ \Omega(\frac{1}{\eps}) \leq k_{\frac12 - \eps} \leq O(\frac{1}{\eps^2})$ by \Cref{thm:near-dom-sets,thm:lb}. We conjecture that the upper bound is correct. 

\paragraph{Approximating Nash equilibria.} We saw that in general, approximate Nash equilibria require support logarithmic in the number of pure strategies, and yet in the particular case of maximal lotteries this blow up can be avoided. Are there other specific types of Nash equilibria that are easy to approximate? While our proof seems to rely strongly on the linear nature of voter preferences, it seems plausible that other types of games may admit similarly advantageous structures.

\anonymize{
\section*{Acknowledgments}
The authors would like to thank Noga Alon for helpful discussions, and for informing us about \cite{bourneuf2025dense}.

Moses Charikar is supported by a Simons Investigator Award. Prasanna Ramakrishnan is supported by Moses Charikar’s Simons Investigator Award and Li-Yang Tan’s NSF awards 1942123, 2211237, 2224246, Sloan Research Fellowship, and Google Research Scholar Award.
}

\bibliography{ref}
\bibliographystyle{alpha}

\appendix

\section{A note on lower bounds}\label{sec:lbs}
In this section, we show that the construction of \cite{alon2006dominating} can be used to show $k_{\frac12 - \eps} \geq \Omega(\frac1\eps)$. For our purposes the parameters need to be set slightly differently, but otherwise the construction is the same.

\begin{theorem}\label{thm:lb}
There exist elections such that for any set $S$ of $k$ candidates, there exists $a \notin S$ such that for all $b\in S$
$$\tfrac1n| a\cg b| \geq \frac12 + \Omega\left(\frac{1}{k}\right).$$
\end{theorem}

\begin{proof}
At a high level, we will construct elections where voters correspond to elements, and candidates correspond to pairs $(A, B)$ of disjoint sets of the elements, where $A$ is big ($\Theta(k)$) and $B$ is small ($\Theta(\log k)$). 
Based on the construction, it will be easy to characterize the fraction of voters that prefer some $(A, B)$ over some $(A', B')$ in terms the sizes of intersections between these four sets. Then, given any $k$ candidates $(A_1, B_1),\dots, (A_k, B_k)$, we can find a candidate $(A, B)$ that dominates each of them by a wide margin by making $A$ a superset of each $B_i$, and choosing $B$ so that it has small intersection with each $A_i$.

\vspace{11pt}

In addition to $k$, the construction has three parameters: $a$, $b$, and $t$. We will eventually set 
$$ a = \frac{t}{200}, \quad  b = \log_2 t, \quad k = \frac{t}{200 \log_2 t}$$
for $t$ sufficiently large.

\vspace{11pt}

The candidates in our election are the pairs $(A, B)$ where $A$ and $B$ are disjoint subsets of $[t]$ of size $a$ and $b$ respectively, and we have $n = 2t$ voters $v_1, v_2, \dots, v_t$ and $u_1, u_2, \dots, u_t$. 
The preferences of the voters (\Cref{fig:prefs}) are built on top of an arbitrary ``default'' order $\cg$ over the candidates and its dual (i.e. the reverse of $\cg$).

\begin{itemize}
    \item Each voter $v_j$ has a ranking that agrees with $\cg$, except they shift $(A, B)$ with $j\in A$ to the very top, and shift $(A, B)$ with $j\in B$ immediately below.

    \item Each voter $u_j$ has a ranking that agrees with the dual of $\cg$, except they shift $(A, B)$ with $j\in B$ to the very bottom, and shift $(A, B)$ with $j\in A$ are immediately above.
\end{itemize}  

Note: in each case, the shifted candidates also ordered according to $\cg$ or its dual.

\begin{figure}[!h]
\centering
\begin{tikzpicture}

\draw (0,1.5) -- (8,1.5);
\node[left] at (0,1.5) {$v_j$};

\draw (1,1.5) ellipse (1 and 0.25);
\node[above] at (1,1.8) {$j \in A$};

\draw (2.5,1.5) ellipse (0.5 and 0.25);
\node[above] at (2.5,1.8) {$j \in B$};

\node at (6,1.5) {\large $\succ$};
\node at (1,1.5) {\large $\succ$};
\node at (2.5,1.5) {\large $\succ$};

\draw (0,0) -- (8,0);
\node[left] at (0,0) {$u_j$};

\draw (7.5,0) ellipse (0.5 and 0.25);
\node[above] at (7.5,0.3) {$j \in B$};

\draw (6.0,0) ellipse (1 and 0.25);
\node[above] at (6.0,0.3) {$j \in A$};

\node at (2.5,0) {\large $\prec$};
\node at (7.6,0) {\large $\prec$};
\node at (6.0,0) {\large $\prec$};
\end{tikzpicture}
\caption{A visual of the preferences for voters $v_j$ and $u_j$. Their favorite candidates are on the left, and their least favorite candidates are on the right. }\label{fig:prefs}
\end{figure}

It is not hard to see that $v_j$ and $u_j$ disagree about each pair of candidates $(A, B)$ and $(A', B')$, except if
\begin{itemize}
    \item  $j \in A\cap B'$, in which case both prefer $(A, B)$ over $(A', B')$, or
    \item  $j \in A'\cap B$, in which case both prefer $(A', B')$ over $(A, B)$. 
\end{itemize}

It follows that the fraction of voters that prefer $(A, B)$ over $(A', B')$ is 
$$\frac12 + \frac{|A\cap B'| - |A'\cap B|}{2t}.$$

Now, consider an arbitrary set of $k$ candidates $(A_1, B_1), \dots, (A_k, B_k)$. The goal is to find a candidate $(A, B)$ that dominates all of these by a large margin. In other words, all $A\cap B_i$ are large, and all $A_i \cap B$ are small.

We will simply choose $A$ to contain $B_1, \dots, B_k$ so that $|A\cap B_i| = b$ for all $i$. This is possible since $kb\leq a$. Next, we want to argue that there is a choice of $B$ such that for all $i$, 
$$|A_i \cap B| \leq \frac{b}{2}.$$
(Here, we point out that \cite{alon2006dominating} are only interested in domination, not the margin of domination, so it is sufficient for them to choose $B$ such that $B\setminus A_i$ is nonempty for each $i$. The fact that we make $|A\cap B'| - |A'\cap B| = \Omega(b)$ increases the margin by a factor of $b$, which is where the log factor improvement comes from.)

We will argue that a random subset $B$ of $[t]\setminus A$ of size $b$ satisfies this property with positive probability. For each $i$, the probability that $|A_i \cap B| > \frac{b}{2}$ is at most (very loosely)
$$\frac{\binom{a}{b/2} \binom{t}{b/2}}{\binom{t}{b}}$$
since there are $\binom{t - a}{b} \leq \binom{t}{b}$ choices for $B$, and the number of these with intersection greater than $b/2$ is bounded by the number of ways to choose $b/2$ elements of $A_j$, and any other $b/2$ elements from $[t]$. Using the well-known approximations $(\frac{n}{k})^k \leq \binom{n}{k}\leq (e\frac{n}{k})^k$, the above expression is at most 
$$\frac{(\frac{2ea}{b}\cdot\frac{2et}{b})^{b/2}}{(\frac{t}{b})^b} = \left(\frac{4e^2a}{t}\right)^{b/2} < 2^{-b}.$$
(We see here that 200 is chosen since it is bigger than $16e^2$.) Since $b = \log_2 t > \log_2 k$, we can use a union bound over the $k$ sets $A_1, \dots, A_k$, to conclude that there exists a choice of $B$ as claimed. Finally, we have that for each $i$, 
\[
\tfrac1n|(A, B)\cg (A_i, B_i)| \geq \frac12 + \frac{b}{4t} = \frac12 + \frac{\log_2 t}{4t} = \frac12 + \frac{1}{800k}. \qedhere
\]
\end{proof}

\end{document}